\documentclass[
submission
]{dmtcs-episciences}
\usepackage[fleqn]{amsmath}
\usepackage{amssymb,times,amsthm}
\usepackage[utf8]{inputenc}
\usepackage{graphicx}
\newtheorem{prop}{Proposition}
\newtheorem{thm}{Theorem}
\newtheorem{lem}{Lemma}

\newtheorem{nota}{Notation}
\newtheorem{obs}{Observation}
\title{What Makes the Recognition Problem Hard for Classes Related to Segment and String graphs?}
\author{Irina Musta\textcommabelow t\u a\affiliationmark{1} and Martin Pergel\affiliationmark{2}
\thanks{Partially supported by the Czech Science Foundation grant GA19-08554S.}}
\affiliation{Berlin Institute of Technology, Institut f\"ur Mathematik, MA 2-2, Stra{\ss}e des 17. Juni 136, 10623 Berlin, Germany; funded by Berlin Math. School\\
Department of Software and Computer Science Education (KSVI),
Charles University Prague.}
\begin{document}

\maketitle
\begin{abstract}
We explore what could make recognition of particular intersection-defined classes hard. We focus mainly on unit grid intersection graphs (UGIGs), i.e., intersection graphs of unit-length axis-aligned segments and grid intersection graphs (GIGs, which are defined like UGIGs without unit-length restriction) and string graphs, intersection graphs of arc-connected curves in a plane.

We show that the explored graph classes are NP-hard to recognized even when restricted on graphs with arbitrarily large girth, i.e., length of a shortest cycle. As well, we show that the recognition of these classes remains hard even for graphs with restricted degree (4, 5 and 8 depending on a particular class). For UGIGs we present structural results on the size of a possible representation, too.
\end{abstract}
\section{Introduction}
This article is a full version of an extended abstract \cite{ugigeurocomb}. Here, we provide mainly the details of the constructions promised there.
Geometric intersection graphs are graphs with a geometric representation where each vertex is represented by a geometric object and the adjacency of a pair of vertices corresponds to the fact that the objects have a nonempty intersection. They are a practically important part of graph theory as some generally hard problems become efficiently solvable on some intersection classes (e.g., \cite{ccl}). These efficient algorithms usually require an intersection representation instead of a graph, which motivates the recognition problem, i.e., the question of whether a given graph has a desired intersection representation.

Our attention focuses on unit grid intersection graphs, called UGIG, i.e., intersection graphs of unit-length axis-aligned non-overlapping line segments in the plane. Our results have influence even on other graph-classes, namely on grid intersection graphs that we denote as GIG defined similarly to UGIG, just without the restriction to unit length \cite{Steve,HNZ,UEH}. Grid intersection graphs have been studied also as PURE-2-DIR or B0-VPG. Obviously UGIGs form a subclass of GIGs. GIGs can be generalized to segment graphs, intersection graphs of straight-line segments in the plane. The topological versions of segment graphs are pseudosegment graphs where pseudosegments do not have to be straight, they just have to keep the topological properties of segments \cite{Fel}. Even this class can be generalized to the class of string graphs, intersection graphs of arc-connected curves in the plane \cite{Krat}. Many geometric intersection graph classes were defined in last years. Although it would be worth to mention more graph classes (like, e.g., polygon-circle graphs, graphs of interval filaments, circle graphs and many other), we are focusing already on relatively many classes. Therefore, we try to avoid introducing classes related to our results only losely where it is possible.

As there exist many classes of intersection graphs, it is useful to avoid exploring each class separately. For optimization problems, typically an efficient algorithm for a superclass yields an algorithm for subclasses. For the recognition problem this is not true. Due to this fact, {\em sandwiching} was introduced \cite{KP}. It is an approach close to approximability/inapproximability. A class ${\cal B}$ is sandwiched between classes $\cal A$ and $\cal C$, if ${\cal A}\subseteq {\cal B}\subseteq{\cal C}$. Given a reduction that produces either graphs from $\cal A$ or not even graphs from $\cal C$, this reduction shows hardness for all classes sandwiched between $\cal A$ and $\cal C$.

About the recognition of the mentioned classes, many results are known. The recognition problem of string graphs is NP-complete even for graphs without triangles \cite{SSS,KratII}, for segment and pseudosegment graphs the problem is known to be NP-hard even for graphs with arbitrarily large girth \cite{KP}, it is not known whether the problem is in NP and Cartesian coordinates of endpoints of particular segments in a representation cannot be used as a polynomial certificate of representability \cite{KM}. It is known that GIGs are NP-complete to get recognized \cite{Krat}; the same article shows that any class between 3-DIR (intersection graphs of straight line segments each using one of three permitted directions) and STRING (intersection graphs of arc-connected curves in a plane) is NP-hard to recognize.

Because for many intersection-defined classes the recognition problem is hard, we are trying to explore the reasons. The idea behind graphs without triangles or with high girth is that it is some form of density (of edges) that makes the recognition hard. This holds, e.g., for polygon-circle graphs, not, e.g., segment graphs \cite{KP}. As the "density" tries to restrict the number of edges, we tried to restrict the classes further to get polynomially recognizable classes, namely by bounding the maximum degree (as it efficiently bounds the number of edges in the graph). However, instead of finding polynomially recognizable classes, we show that all explored classes are rather resistant to this criterion.


\section{The Results}

Our main result whose proof we sketch in the following section is:
\begin{thm}\label{thm:iff}
Between the following three pairs of graph-classes no polynomially recognizable class can be sandwiched even when restricted on graphs with arbitrarily large girth and yet on graphs with the maximum degree (at most) 5, 4 and 8, respectively:
\begin{itemize}
\item UGIG and pseudosegment graphs,
\item GIG and segment graphs,
\item GIG and string graphs.
\end{itemize}
\label{thm:main}
\end{thm}

Conversely, we may state an observation that these graph classes are polynomially (or even trivially) recognizable for graphs with maximum degree 2, as such graphs consist of disjoint union of paths, cycles and isolated vertices. All such graphs either have the appropriate representation (by a class containing at least 3-DIR, i.e., class representable by segments in 3 directions), or, for UGIG, it might be necessary to verify whether the graph is bipartite. Provided that the input graph is bipartite for UGIGs, the recognition algorithm may just answer "yes". 

Note that between UGIG and GIG a sophisticated (hard) class can be established, e.g., length of an odd cycle can encode a program for universal Turing machine and the class contains just cycles encoding the Halting problem making such class algorithmically not recognizable. So, for classes sandwiched between UGIG and 3-DIR we must be careful even for graphs with maximum degree 2.

\begin{table}
\centering
{\bf Overview of the results}
\begin{tabular}{l l l l l l}
Class/Degree & 2 & 3 & 4 & 5, 6, 7 & 8\cr
UGIG & P & ? & ? & NP-complete & NP-complete\cr
GIG & P & ? & NP-complete & NP-complete & NP-complete\cr
... & any & ? & NP-hard & NP-hard & NP-hard\cr
SEG & P/trivial & ? & NP-hard & NP-hard & NP-hard\cr
... & P/trivial & ? & NP-hard & NP-hard & NP-hard\cr
PSEG & P/trivial & ? & NP-complete & NP-complete & NP-complete\cr
... & P/trivial & ? & ? & ? & NP-hard\cr
string & P/trivial & ? & ? & ? & NP-complete\cr
	string3 & P/trivial & ? & {\it NP-complete\cite{Krat}} & {\it NP-complete} &{\it NP-complete}\cr
\end{tabular}
\caption{Results in italics are not our ones. The class string3 are string graphs without the restriction on girth, i.e., we permit even triangles. Other results hold even for graphs with arbitrarily large (constant) girth. Three dots as the class-name mean any class sandwiched between the neighboring classes. As the class may be created artificially, it can be arbitrarily hard. Therefore, we can claim something only about its hardness (or triviality).}
\end{table}

Our results answer a question of Kratochv\'\i l and Pergel \cite{KP} whether string graphs with sufficiently large girth can be recognized in a polynomial time. Our construction shows not only that large girth does not help. Moreover, we show that even graphs with relatively low maximum degree are hard to get recognized.

In the rest of this section we show upper and lower bound on the size of a UGIG representation:
\begin{figure}[htbp]
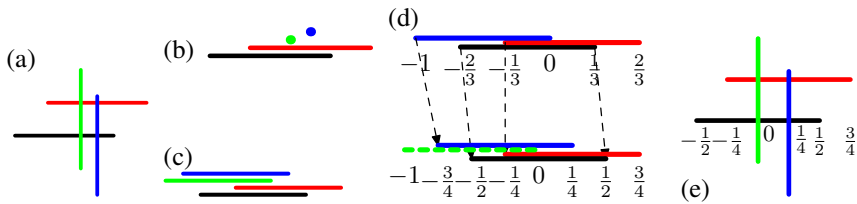

\hbox{\vbox{\hsize=2cm (a)\vskip -1mm ~\includegraphics[width=1.8cm]{obrazky.6}}  \hfill
 \vbox{\hsize=2.9cm  (b) ~\includegraphics[width=2.2cm]{obrazky.7}
\vskip 1cm
(c) ~\includegraphics[width=2.4cm]{obrazky.8}} \hfill 
\vbox{\hsize=3.8cm (d) 
\includegraphics[width=3.5cm]{obrazky.9}}
\hfill \vbox{\hsize=2.4cm (e)\hskip -4mm 
\includegraphics[width=2.4cm]{obrazky.11}}}~
\caption{Considering the arrangement from picture (a) we make a projection onto one axis (b), i.e. perpendicular segments collapse into single points. We extend them in one direction (in our case to the left), see picture (c). Picture (d) shows the last iteration of the fourth step of the sweeping algorithm, i.e. extending the so far obtained representation. Picture (e) shows the final representation (after having also swept along the $y$-axis)}
\label {fig:sweep}
\end{figure}

\begin{prop}
Any unit grid intersection graph on $n$ vertices can be represented in a grid $(n+1)\times (n+1)$ with all coordinates being multiples of $\frac{1}{n}$. Moreover, we can find such a representation that with respect to each axis no two segments have the same non-integral part of the coordinate.
\label{prop:canon}
\end{prop}

Note that the lower bound for the granularity (multiples of $\frac{1}{n}$) is tight because of $K_{1,n-1}$ which requires this precision for the (distinct) coordinates, as all $n-1$ segments of one part need to intersect (not just touch) the same segment of the other part and it is prohibited that two segments start on the same coordinate.

\begin{proof}
Without loss of generality, in the sequel we will describe  the procedure for the $x$-axis. Projecting this arbitrary representation of $G$ on the $x$-axis, we obtain a sequence of intervals (corresponding to the vertices in $H$) and points (corresponding to the vertices in $V$). By eventually employing small perturbations, it can be assumed that the projected elements are all in general position: no interval is degenerated, no two endpoints coincide (here we treat the projection points stemming from $V$ as endpoints). 

To simplify the following sweeping-argument, we extend individual points to unit-segments (in an arbitrary fixed way along the $x$ axis). Thus we obtain an arrangement of unit segments. We build the canonical representation by performing a sweep from  left to  right and employing the following steps:

\begin{itemize}
\item For the left-most segment we assign to the right endpoint the coordinate $0$ (hence the left endpoint now has coordinate $-1$). This segment is now the reference from which the rest of the construction will follow.
\item We process the segments in their left to right ordering where each newly added segment  has its left endpoint assigned to the next free 'slot': namely, the smallest multiple of $\frac{1}{n}$ that is free, such that the overlapping (or disjointness) conditions to the previously assigned segments are not violated. This is illustrated in  Figure \ref{fig:sweep}. Note that due to the chosen granularity it is always possible to find such a free slot. In the worst case, i.e. the segment to be added is disjoint to all the previous ones, the needed increment for the left coordinate is $1+\frac{1}{n}$. This leads to an upper bound of $n+1$ large bounding size per coordinate (which the initially mentioned example fulfills, note however, this is not an optimal representation).
\item In the next step, the projection on the $y-coordinate$ is processed analogously.
\item After the double sweep, all relevant coordinates being known, one can draw the final configuration. Note that the above procedure does not change the relative position of the endpoints and projected segments, hence not contradicting the adjacency and non-adjacency conditions. 
\end{itemize}
\end{proof}

\begin{figure}[htbp]
\centering
\includegraphics[width=8cm]{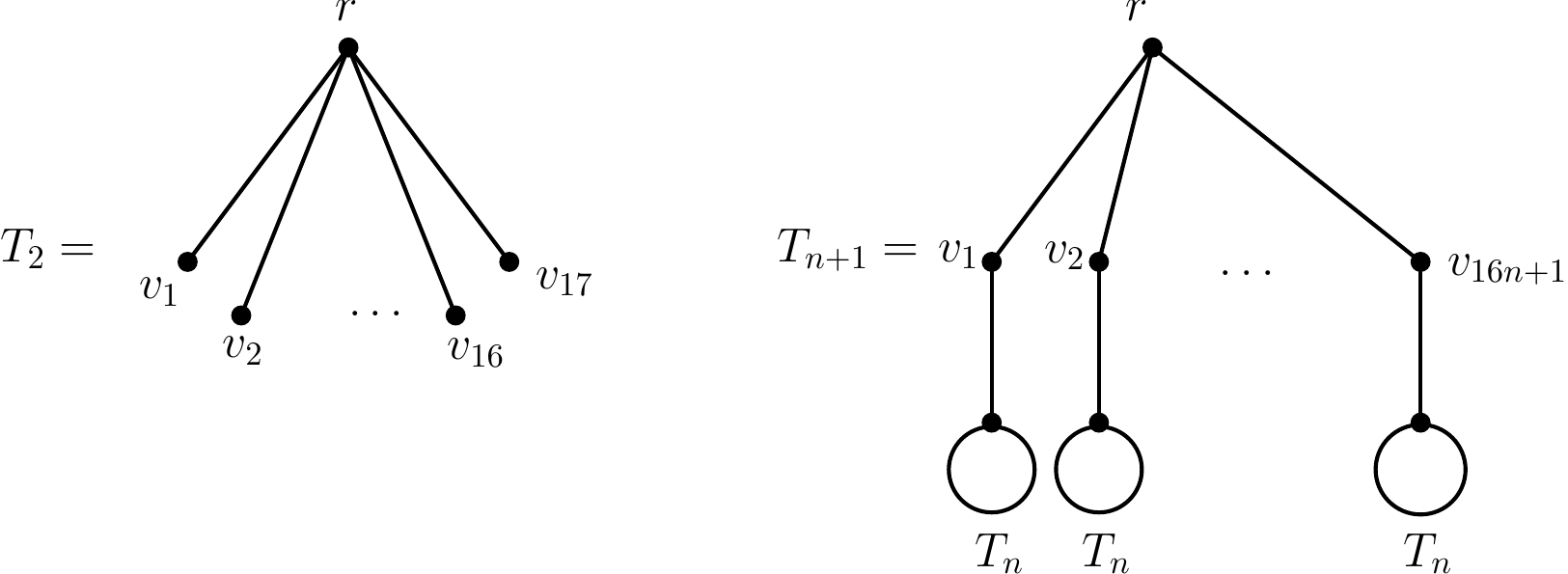}
\caption{Constructing the family of trees}
\label{tfamily}
\end {figure}

In the sequel, we define the \emph{boundary size} as the semiperimeter of the bounding rectangle and prove that:

\begin{thm} \label{thm:treeunbound}
 For all $n\geq 2$, a UGIG representation of $T_n$ needs a boundary size of at least $n$.
 \end{thm}

We proceed inductively:

The base case $n=2$ is clear, as at least one unit is needed in both the horizontal and the vertical direction.

  For the induction step, we apply the pigeonhole principle several times. Out of the $16n+1$ children of the root, at least $8n+1$ have their children either all above the root, or all below. Without loss of generality, we assume the latter is the case. Out of these $8n+1$ nodes, either at least $4n+1$ have children with an endpoint to the left of the root, or $4n+1$ have children to the right of the root. W.l.o.g. we consider the latter case. Restricting our analysis to these $4n+1$ nodes and their successors, we see that the paths of length two descending from the root form a nested structure, similar to the one in Figure \ref{selvert}. 

\begin{figure}[htbp]
\centering
\includegraphics[width=4.5cm]{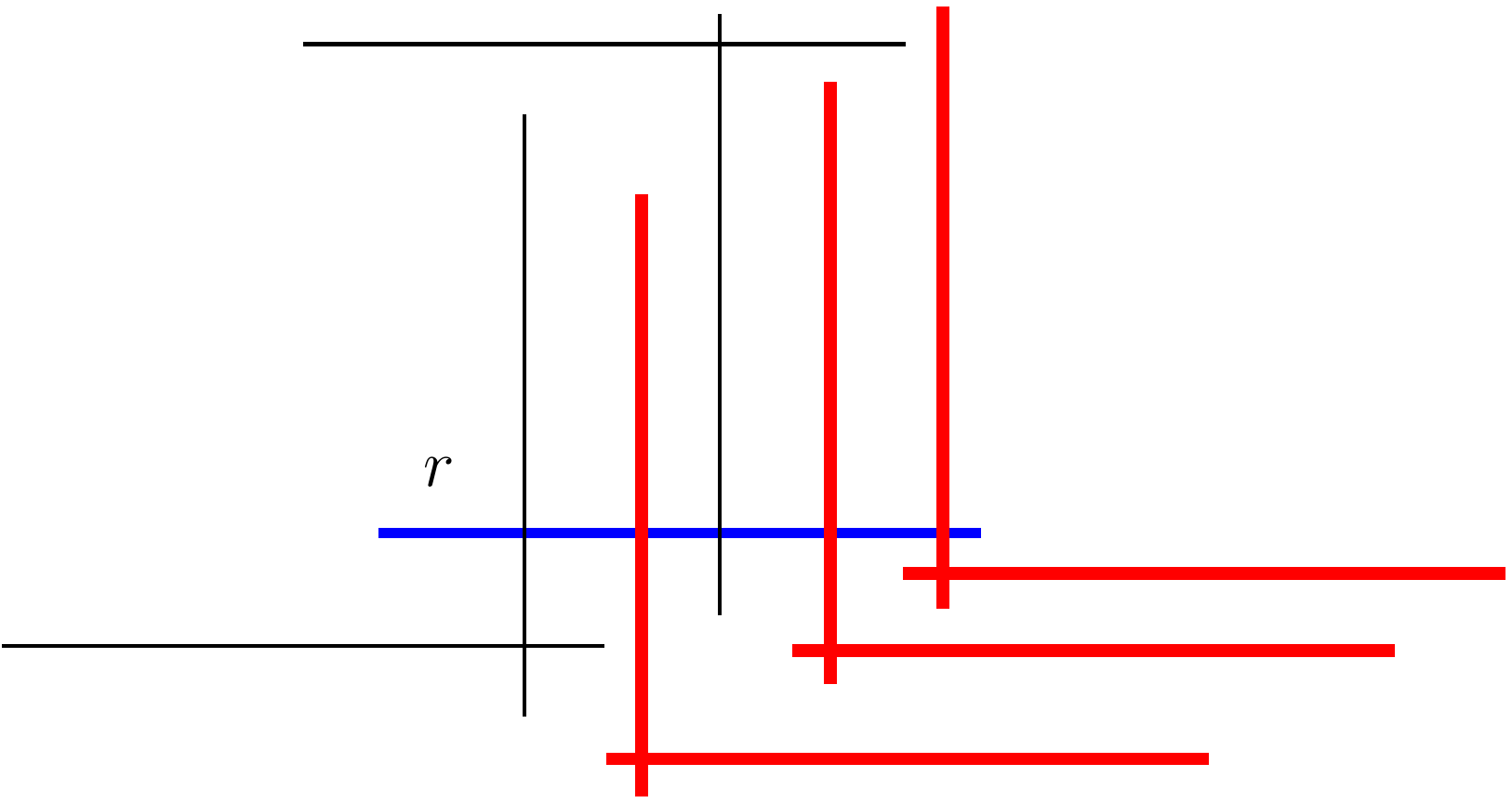}
\caption{Selecting the nested paths}
\label{selvert}
\end {figure}

Consider the second lowest lying child of the root and let $l$ be its descendant. We distinguish two cases:

\begin{itemize}
\item At least one of the attached $T_n$s  does not go above  $l$. Then, from the induction hypothesis, we need a boundary size of at least $n$ for the copy of $T_n$ and an extra unit in the vertical direction that lies above $l$. Hence, in total, a rectangle of semiperimeter at least $n+1$ is necessary.
\item All attached $T_n$s go above $l$. This can only happen via a vertical segment to the right or to the left of $l$. We apply the pigeonhole principle again, to conclude there exist $2n+1$ root children whose descending $T_n$s reach above $l$ from the right (the other case is analogous).
\end{itemize}

\begin{lem}
Let $u_1, \ldots u_{2n+1}$ be the top-down ordering of the descendants of the above $2n+1$vertices and $P_1, \ldots P_{2n+1}$ the corresponding paths that go above $l$. Then $y_{min} (horizontal(P_{2k})) < y_{min} (P_{2k-1})$  and $y_{min} (P_{2k+1}) < y_{min} (P_{2k})-1$ for all $k\geq 1$ (here $y_{min}(P_i)$ denotes the smallest $y$ coordinate of the path $P_i$). 
\end{lem}

\begin{proof}
Note that at every even step, a horizontal segment is added strictly below the current lowest vertical one.

The second statement follows from the fact that the currently lowest horizontal segment can only be passed by a path from left to right, if this path has a vertical segment lying strictly below, otherwise the unit segment condition would be violated. The previous step ensures the depth increasing by 1.
\end{proof}

As a corollary of Proposition~\ref{prop:canon}, Theorem~\ref{thm:main} and a trivial observation, classes UGIG and GIG are NP-complete to recognize even with arbitrarily large girth.
\section{The Reductions}
\label{ssec:gadget}
First, we sketch the proof of Theorem~\ref{thm:main}, next we provide the details of the construction. We reduce planar 3-connected 3-SAT(4) shown to be NP-complete \cite{KratII}. This is a special version of 3-SAT where each variable occurs at most 4 times and the incidence graph (of the formula) is planar and 3-connected (3-connectivity implies a unique planar embedding up to the outer face). The incidence graph is bipartite with one part formed by variables, the other by clauses and an edge means a presence of a variable in a clause. We follow the ideas of \cite{KP,FMP}, i.e., for a planar embedding of the incidence graph, vertex representatives get represented by variable-gadgets or truth-splitters, clause representatives we replace by clause gadgets and the edges (of the incidence graph) we replace by pairs of paths whose left-right orientation represent the truth-assignment.
	
Variable-gadgets must keep the occurrences synchronized, clause-gadgets must be representable exactly for 7 satisfiable assignments. For the 1st case, as a variable gadget we take two vertices connected by an edge and each pair representing an occurrence stems from them to incur the situation of Figure~\ref{fig:varass}. A clause-gadget is depicted in Figure~\ref{fig:clausevar1}, dashed curves depict arbitrarily long paths. For the 2nd case we use the same clause-gadget, but to decrease the maximum degree in the variable-gadget, we stick the 1st and the 2nd occurrence together and we represent the 3rd and 4th by the truth-splitter as depicted in Figure~\ref{fig:truthsplitter} without vertex $c$.  For the 3rd case, we use the truth-splitter in Figure~\ref{fig:truthsplitter} (including vertex $c$) and a clause-gadget from Figure~\ref{fig:stringclause}.

\begin{figure}[htbp]
\centering
\includegraphics[width=5.5cm]{obrazky.22}
\includegraphics[width=6.5cm]{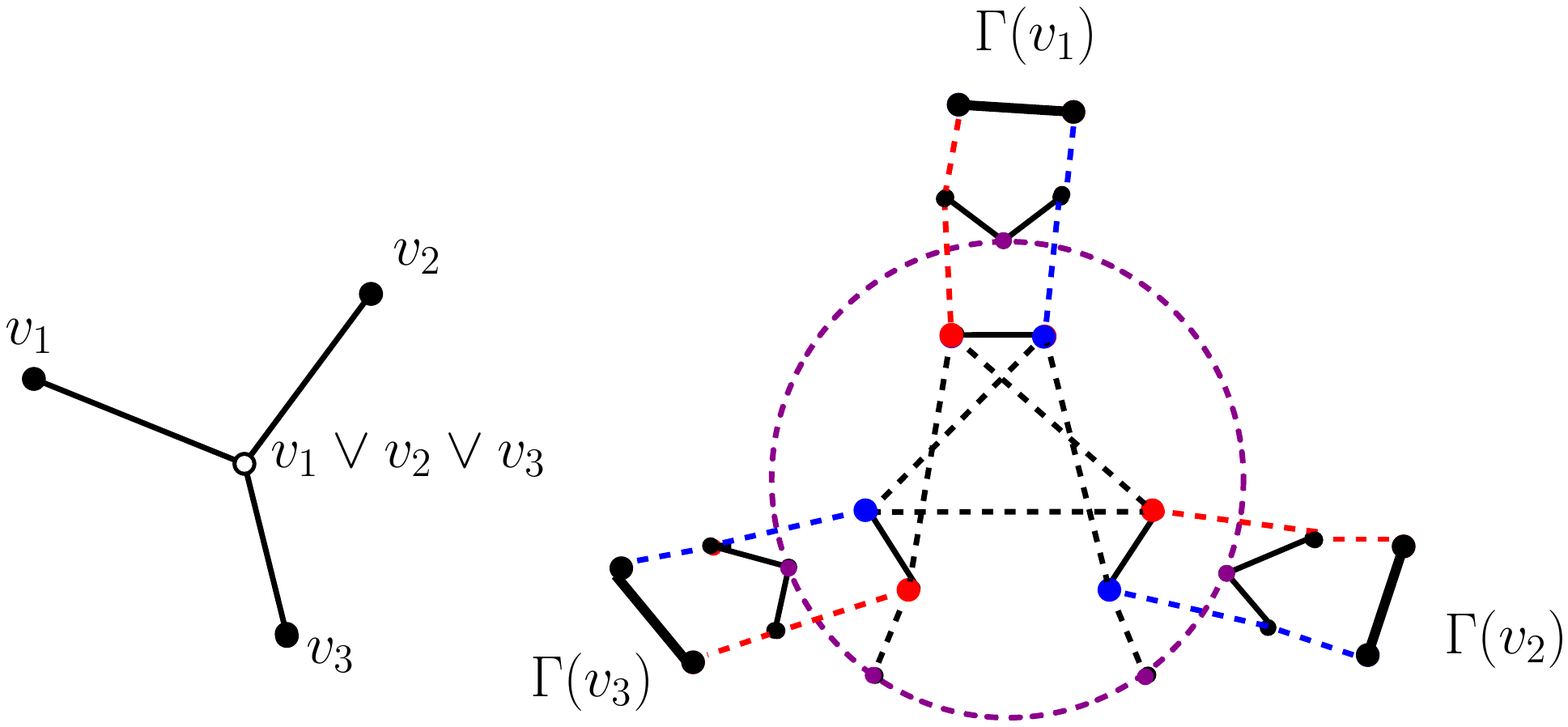}
\caption{Left: Incidence graph for a formula $(v_1\vee v_2\vee v_3)\wedge(\neg v_2\vee v_3\vee\neg v_4)\wedge (\neg v_1\vee v_2\vee\neg v_4)\wedge(\neg v_1\vee\neg v_3\vee v_4)$ (here we have identified the vertex names with the variables/clauses). Middle: Subgraph of the incidence graph for $v_1\vee v_2\vee v_3$. Right: The corresponding subgraph whose representability we explore. This subgraph shows the part inside the red dashed triangle in the left picture. Left-right orientation of blue and red vertices/paths determines the truth assignment
}

\label{fig:clausevar1}
\end {figure}

\begin{figure}[htbp]
\centering
\includegraphics[width=10cm]{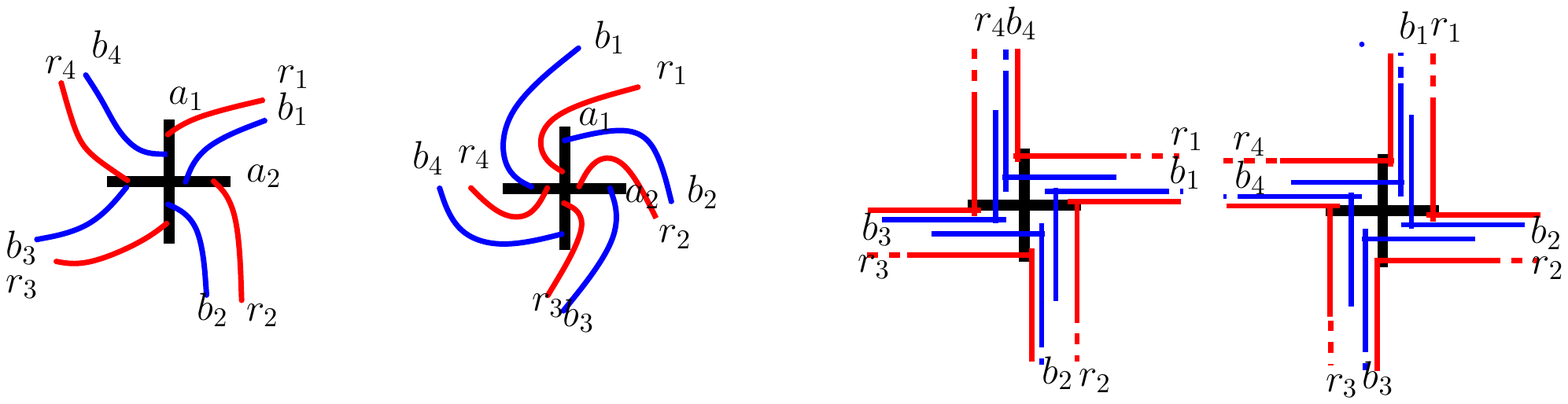}
\caption{The two distinct representations of the four occurrences incident to a variable gadget. The left pair is a pseudosegment representation, the right one a UGIG representation. In each pair, the left corresponds to the FALSE assignment, the right one to the value TRUE}
\label{fig:varass}
\end{figure}

\begin{figure}
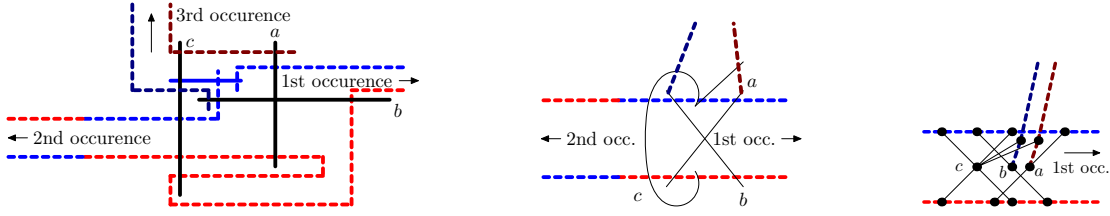

\includegraphics[width=5.5cm]{strings.1}\hfill\includegraphics[width=3.5cm]{strings.2}\hfill \includegraphics[width=2.5cm]{strings.3}~
\caption{Truth-splitter gadget. First occurrence always goes to the right, second to the left, third upwards. Left: GIG representation, middle: string representation, right: the appropriate graph}
\label{fig:truthsplitter}
\end{figure}

\subsubsection*{The gadgets}
We first describe the \textbf{clause gadget}, which is the same as in \cite{KP}. We show it on the right in Figure \ref{fig:clausevar1}.


In the Figure \ref{fig:clausevar1} (right side), straight segments represent one edge, dashed lines depict arbitrarily long paths whose length depends on the required girth. Furthermore, any two paths that are not explicitely listed as sharing a vertex do not intersect.

The three pairs of adjacent vertices inside the gadget correspond to the three variables of the clause, and the pairs of paths incident to them correspond to the occurrence gadgets, which we describe in more detail below. To explain the construction in Figure \ref{fig:clausevar1}, we denote the variables as 1st, 2nd and 3rd  in clockwise order, with the 1st variable being represented by the topmost adjacent pair of vertices.

 We color these six vertices red and blue, as follows: 

\begin{itemize}
\item Each adjacent pair consists of a blue and a red vertex.
\item For the third pair, the vertex incident to the path connected to the cycle is colored red if and only if the corresponding variable does not appear negated  in the clause.
\item For the second pair, the opposite is the case.
\item For the first (topmost) pair, consider the vertex $v$ connected to the vertex of the leftmost pair which is red if not negated. Then $v$ is colored red if and only if the corresponding literal is positive. 
\end{itemize}

To obtain the circular orientation of the variables in the clause gadget, consider the unique planar embedding of $\mathcal{B}(\mathcal{F})$, and pick the circular ordering of the edges connecting the clause vertex to its three variable vertices. Moreover, to adjust for the desired girth, the lengths of the paths inside the gadget, as well as the length of the cycle between two different pairs, can be chosen arbitrarily long. We further adjust the path lengths to ensure the clause subgraph is bipartite and that the red vertices belong to the same partition.

The \textbf{variable gadget} is simple: it consists of two adjacent vertices, which we denote by $a_1$ and $a_2$ in the sequel.

The \textbf{occurrence gadget} consists of a pair of non-intersecting paths connecting the pair of vertices of the variable gadget to the corresponding pair of vertices in the clause gadget. To determine the exact adjacencies, we assume the variable gadget of $x_1$ has all four occurrences (less then three is impossible, due to $\mathcal{B}(\mathcal{F})$ being 3-connected). Let now $C_1,C_2,C_3,C_4$ denote the clockwise circular order of the clauses that contain $x_1$ as a variable. Then, we connect $a_1$ to the red vertices of $C_1$ and $C_3$ and to the blue vertices of $C_2$ and $C_4$. Analogously, we connect $a_2$ to the blue vertices of $C_1$ and $C_3$ and to the red vertices of $C_2$ and $C_4$. As in the case of the clause gadget, these paths can be chosen long enough to ensure $G(\mathcal{F})$ has the desired girth and is bipartite. Note that since the plane embedding is rigid, the lengths of the edges of $\mathcal{B}(\mathcal{F})$ are fixed. Therefore, in $G(\mathcal{F})$, the lengths of the occurrence paths are chosen proportional to the length of the corresponding embedded edge of $\mathcal{B}(\mathcal{F})$.

\begin{nota}
In the sequel, we refer to paths connecting a red vertex in a clause gadget to a vertex in a variable gadget as to \emph{red paths}; moreover, we use the same description for paths connecting two red vertices. The definition of \emph{blue} paths is analogous.
\end{nota}

In the case only three clauses contain $x_1$ as a variable, we introduce a \textbf{dummy occurrence}. The dummy occurrence consists of a pair of paths that start from the vertices $a_1$ and $a_2$ of the variable gadget and  are connected to two other occurrence gadgets via another path, as in the Figure \ref{fig:dummy}. This latter path ensures that due to 3-connectivity in every plane pseudosegment representation of $G(\mathcal{F})$, the dummy occurrence is "fixed" between the two other occurrences to which it is connected.

\begin{figure}[htbp]
\centering
\includegraphics[width=8cm]{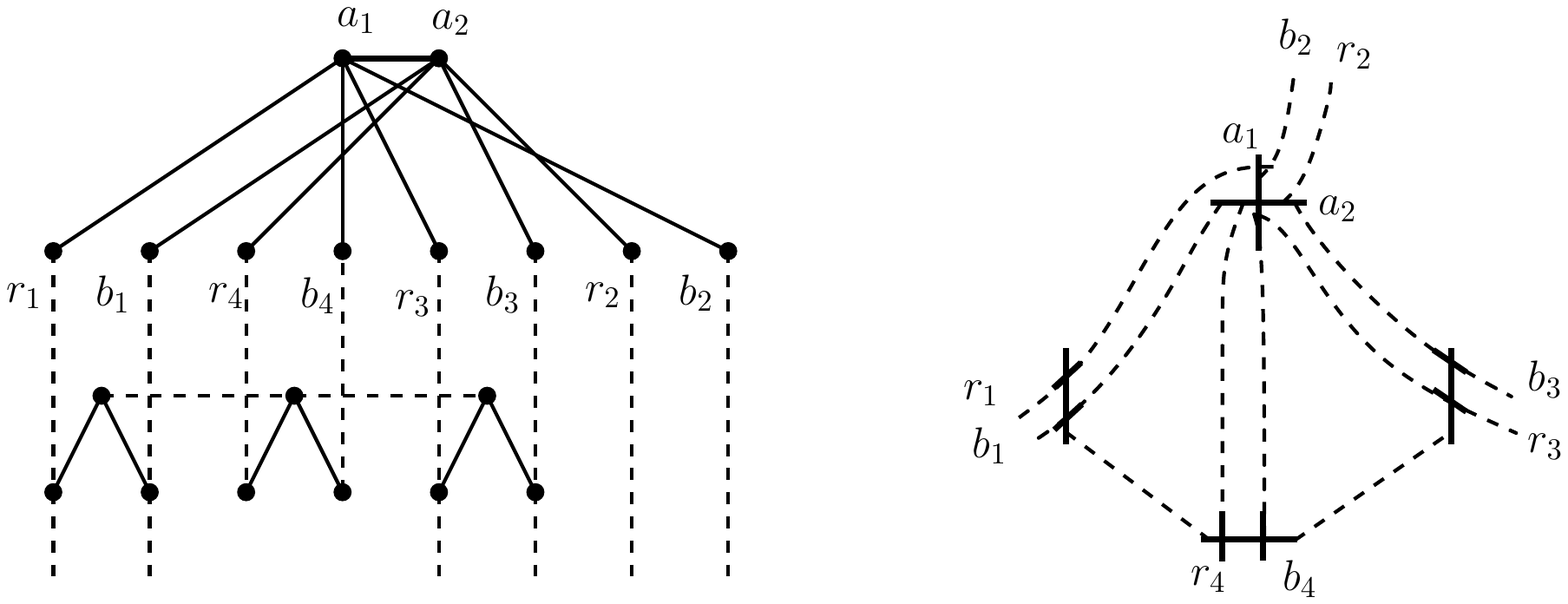}
\caption{Left: The dummy occurrence, consisting of the paths $r_4,b_4$ and the path that connects them to occurrences 1 and 3. Right: A segment representation of the dummy occurrence. For occurrence $i$, its red path is represented by $r_i$, its blue one by $b_i$.}
\label{fig:dummy}
\end {figure}

\subsection*{Variable assignments and the UGIG representability of $G(\mathcal{F})$ }

After having described the gadgets and thus the structure of $G(\mathcal{F})$, we proceed to show that:

\begin{thm}
 $G(\mathcal{F})$ is representable as a UGIG if and only if $\mathcal{F}$ is satisfiable. Moreover, if the latter is not the case, $G(\mathcal{F})$ cannot be represented even as an intersection graph of pseudosegments. 
\end{thm}

\begin{proof}
We start with a series of observations, concerning a pseudosegment representation of $G(\mathcal{F})$:

\begin{obs} \label{obs:fixorder}
Each of the four pairs of occurrence paths incident to a variable gadget blocks the visibility of the intersection point of the segments of $a_1$ and $a_2$ from exactly one quadrant (otherwise, we either need less occurrences, or obtain path intersections, a contradiction). The uniqueness of the planar embedding of the contracted graph $\mathcal{B}(\mathcal{F})$ fixes the clockwise ordering of the occurrences. Therefore, changing the clockwise order of the two paths in an occurrence, forces the same change in the representation of all of the other three occurrences. 
\end{obs}


\begin{obs} 
 All intersections between occurrence paths belonging to different clauses can be avoided. Indeed, if such intersections are unavoidable, they remain so even after contracting the representation of $G(\mathcal{F})$ back to the one of $\mathcal{B}(\mathcal{F})$. However, the hypothesis states that the latter is plane. Hence, given a variable-consistent clockwise ordering of the paths in an occurrence, a pseudosegment representation is impossible if and only if there exists a clause gadget which becomes impossible to represent.
\end{obs}

\begin{obs} \label{obs:allpos}
We can assume without loss of generality that all literals are positive. Indeed, any UGIG representation of a clause gadget containing a negated literal is equivalent to an UGIG representation of the same clause gadget with all literals positive, but opposite variable assignment for all originally negative ones. This corresponds in practice to swapping the blue and red colors in the corresponding pair of crossing segments (and occurrence paths, respectively).
\end{obs}

$"\Rightarrow:"$ 
Assume $\mathcal{F}$  is satisfiable and fix a variable assignment that evaluates to TRUE. For each variable gadget, choose the ordering of the paths in the occurrence pairs such that the red path is clockwise before the blue path if and only if the value of the variable is assigned to true. First, note that the variable and occurrence gadgets admit unit grid intersection representations in all cases. Together with first two observations above, it follows that all we need to check  now is that the clause gadgets also admit unit grid intersection representations compatible with the representations of the occurrence gadgets. That is to say, the three pairs of occurrence paths leave it in a consistent order to that prescribed by their entering the variable gadgets.

 For each clause, there exist seven combinations  of variable assignments that satisfy it.  Each such variable assignment corresponds to an ordering of the red and blue occurrence paths, within each of the three pairs. By our construction and using Observation \ref{obs:allpos}, a literal evaluates to true if and only if the red occurrence path appears at the left of the blue  path. The Figure \ref{fig:allatonce1} shows how it is possible to create a unit grid representation in the seven cases when the clause gadget evaluates to true, and none of the variables is negated.

\begin{figure}
\includegraphics[width=12cm]{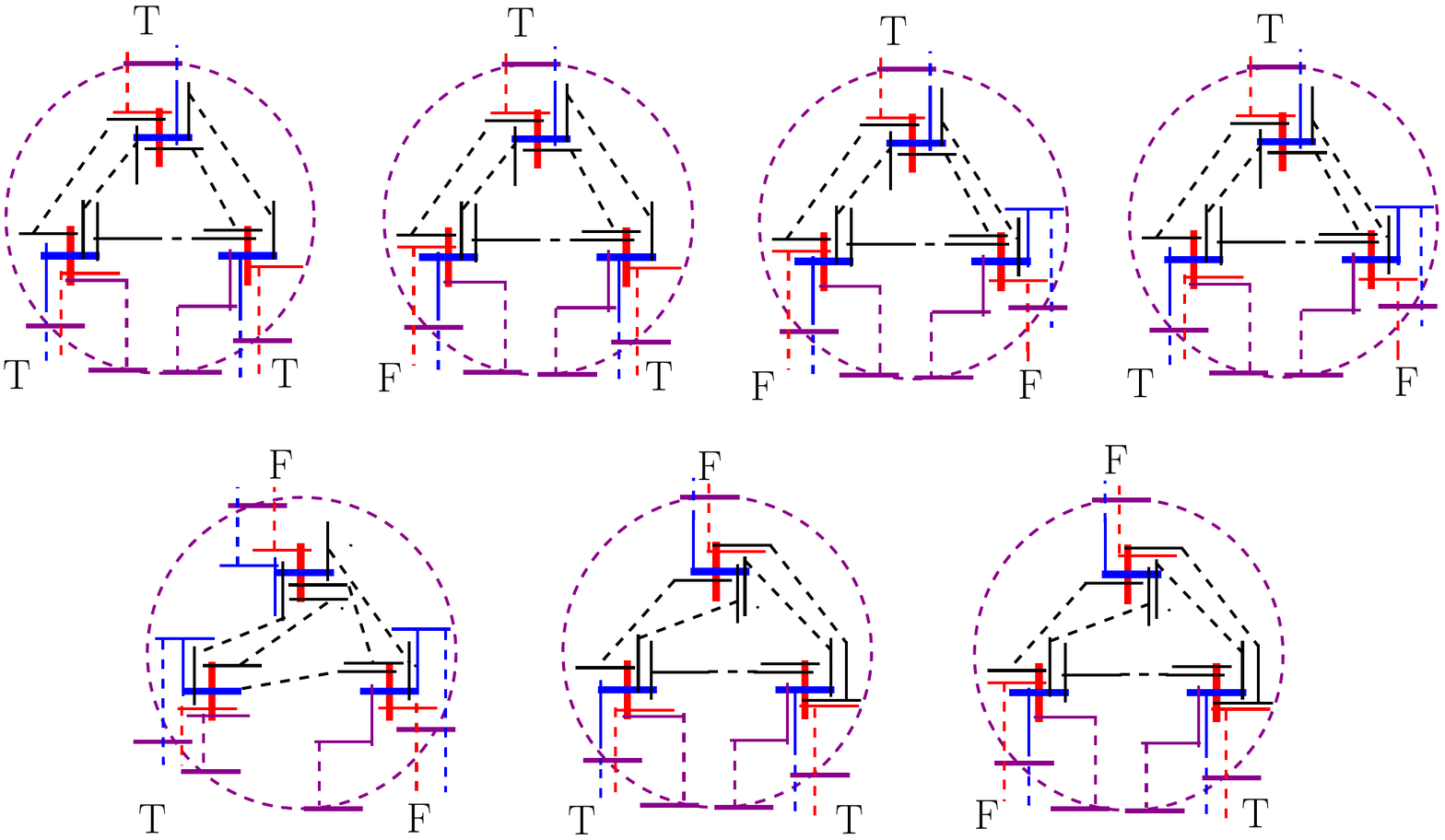}
\caption{The clause gadgets admit UGIG representations for all variable assignments that make the clause evaluate to true. The violet solid segments correspond to those segments of surrounding circles where the occurrence gadgets are leaving the clause gadgets}
\label{fig:allatonce1}
\end {figure}

$"\Leftarrow:"$ Assume $ \mathcal{F}$ is not satisfiable, but $G(\mathcal{F})$ admits a pseudosegment representation. From Observation \ref{obs:fixorder} it follows that for each variable gadget, either all the red occurrences are clockwise before the corresponding blue ones, or all clockwise after. In this representation, there necessarily exists a clause where all the blue paths are represented clockwise before the paired red ones (after swapping colors for the negated literals, cf. Observation \ref{obs:allpos}). Indeed, if this is not the case, then, by assigning the values of variables like in the proof of the forward direction, we obtain a truth assignment for $ \mathcal{F}$ that makes it satisfiable, a contradiction.

Consider now the clause $C$ where for all three variables the blue paths must occur clockwise before the red paths, as sketched in Figure \ref{fig:notposs1}. Recall that in this representation we assign to each literal a pair of intersecting pseudosegments. Starting from the top and proceeding clockwise, we denote these literals by $L_1,L_2,L_3$, such that the two anchoring paths start from the red pseudosegment of $L_3$ and the blue one of $L_2$. In the sequel, we will abuse the notation and denote by $L_i$ both the literal and the corresponding pair of pseudosegments.

First, we consider the red-blue intersecting pair that corresponds to $L_3$. The proper intersection of two pseudosegments creates four half-pseudosegments sharing the intersection point as an endpoint. We label these half-segments clockwise by I, II, III, IV, such that in this representation the blue occurrence is incident in IV. Then, the path connecting the red pseudosegment to the cycle must be incident to III: incidence to I implies, if the path order is preserved, that the blue segment becomes isolated from those belonging to the other two crossing pairs, a contradiction. For the same reason, the red occurrence path is incident in I. Symmetrically, we can label the segment halves for the  crossing pair of $L_2$ such that the path connecting the pair to the cycle is incident in III, the red occurrence in II and the blue occurrence in I.

\begin{figure}[htbp]
\centering
\includegraphics[width=4.5cm]{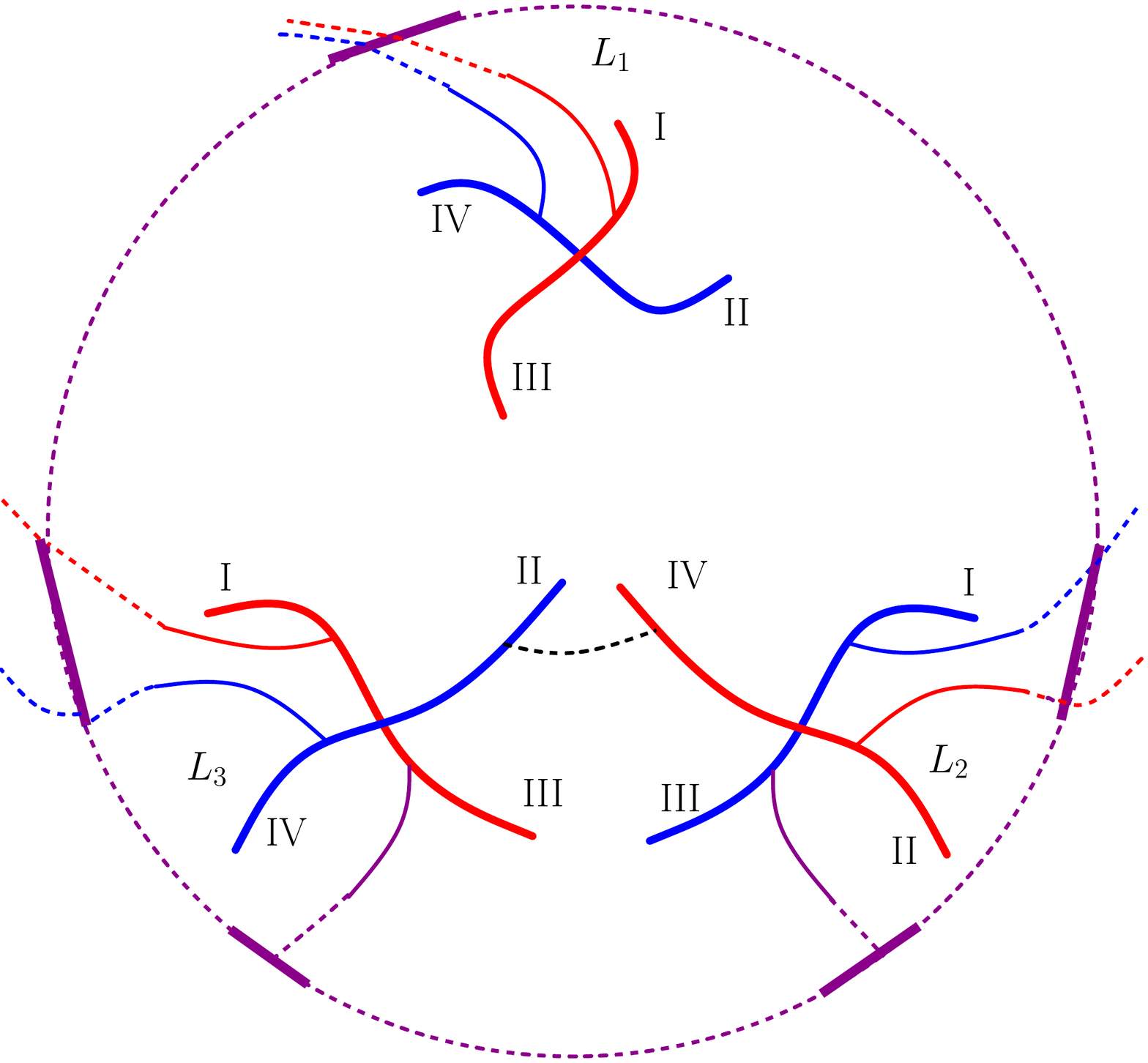}
\caption{The placement of the occurrence and anchoring paths with respect to the half-segments of the crossings in an unsatisfied clause gadget}
\label{fig:notposs1}
\end {figure}

Now, in order to connect the red  pseudosegment of $L_2$ with the blue one of $L_3$, this path needs to start in II of $L_2$ and end in IV of $L_3$, since the red II of $L_3$ and blue IV of $L_2$ have no visibility. Then, the path connecting the blue pseudosegment of $L_3$ to the red one of $L_2$ must have III-III as start and ending points, since the only other option, I-I, isolates the blue $L_3$ and red $L_2$ pseudosegment from the crossing pair of $L_1$, a contradiction. The paths are also represented in Figure \ref{fig:notposs1} above. We now need to add four more connecting paths, ending up in one of the following cases:

 \begin{itemize}
 \item If the $L_1 -L_3$ red path starts in the I half-segment of  $L_1$ and ends in the I half-segment of $L_3$, then the $L_1$ blue pseudosegment is isolated from  $L_2$.
 \item If the $L_1 -L_2$ red path starts in the I half-segment of  $L_1$ and ends in the IV half-segment of $L_3$, then the $L_2$ blue pseudosegment is isolated from  $L_3$.
 \item If the $L_1 -L_3$ red path goes from III to I and the $L_1 -L_2$ red path from III to IV, then the $L_2$ blue pseudosegment is isolated from $L_1$. 
 \end{itemize}

 We can therefore conclude that whenever $\mathcal{F}$ is not satisfiable, $G(\mathcal{F})$ does not admit a pseudosegment representation.

\end{proof}

For the degree in the construction, we can see that the highest degree occurs in the variable gadget where we need vertices of degree 5. In the next section we briefly mention how to get rid of this problem.

\subsection*{GIG}

We use suitable parts of the previous construction, i.e., we just change the variable gadget (for a gadget with lower degree). In order to represent occurrences of a variable (that are consistent with the truth-assignment) we pick the occurrence gadgets or 1st and 2rd occurrence (pairs of paths) and stick them together. We label the occurrences so that the remaining is the 3rd, if there remain two (3rd and 4th), the 3rd is above paths representing 1st and 2nd while the 4th is below. If we stick the paths (representing 1st and 2nd occurrence) like shown in Figure~\ref{fig:truthsplitter}, right picture without vertex $c$. In this way we keep them synchronized. The depicted construction enforces the 3rd occurrence (steming upwards from this pattern) to be synchronized with the 1st occurrence. For the 4th occurrence we perform the same split from the 2nd occurrence. Vertex $c$ is so called {\em degradation gadget} that is used in the next section. This construction uses only vertices of maximum degree 4.

This implies the 2nd result, i.e.:
\begin{thm}
Between class of grid intersection graphs and  pseudosegment graphs no polynomially recognizable class can be sandwiched even if we restrict on graphs with arbitrarily large girth and maximum degree at most 4.
\end{thm}

\subsection*{String graphs}

\begin{thm}
Between class of grid intersection graphs and string graphs no polynomially recognizable class can be sandwiched even if we restrict on graphs with arbitrarily large girth and maximum degree at most 8.
\end{thm}

\begin{proof}
Figure~\ref{fig:truthsplitter} shows GIG representation of the variable gadget and of truth-splitter. Occurence gadgets (pairs of paths) are represented in the same way as for UGIGs, so the only problem is to show the representability of clause gadget for satisfiable assignments and irepresentability for unsatisfiable assignment. The former is just an exercise either on homeomorphism if we try to stretch and align representation in Figure~\ref{fig:stringclause}, or we may extend Figure~\ref{fig:allatonce1}.

The remaining part is to explain why non-representable version of the clause gadget admits not even a string representation. In the text, we promise that strings depicted on the left picture in Figure~\ref{fig:stringclause} by straight green segments cannot make use of further (mutual) intersections. First observation is that due to the design of the gadget, they cannot be represented outside the dashed black boxes (drawn suggestively around their intersection-point). Each green segment intersects the box in two points (i.e., it intersects two curves representing its boundary). When we try to find a path (for a possible string) that would connect these two points, we realize that these two points are separated by a Jordan curve (formed by curves that must not be intersected with the green string) and therefore the green string has to pass through the box. That is why our drawing with the intersection of green segments inside the black dashed box is not misleading. When we pick any pair of these green segments and try to make the appropriate curves to intersect outside the black dashed square, again, we find a Jordan circle that cannot be crossed by either (green) curve and therefore they cannot mutually intersect outside the black dashed box (drawn on the picture around their intersection). That is why we turn up in a similar situation as is depicted in Figure~\ref{fig:notposs1}; in fact, around the intersections of red and blue curves, there is the representation of the black (dashed) box and this box must not be accessed by any path present in the segment version of the gadget. Therefore, the paths connecting red and blue vertices (in Figure~\ref{fig:notposs1}) can use just four "half-strings" (that mutually do not intersect) and because of the same reason as in the case of pseudosegments they cannot be represented.

Maximum degree in the construction is 8.
\end{proof}

\begin{figure}
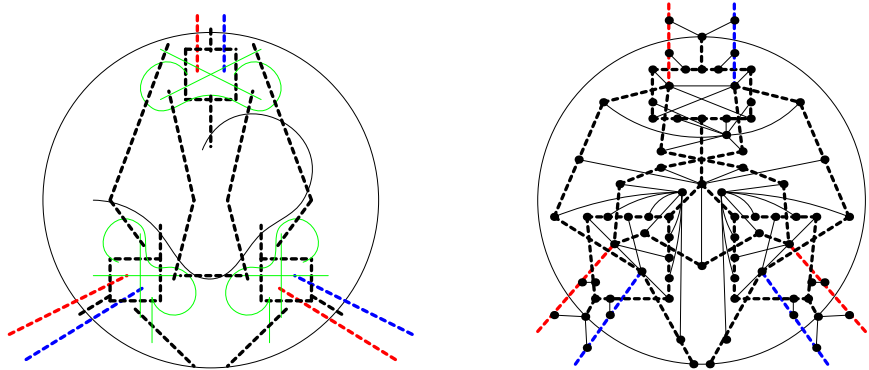

\hfill\includegraphics[width=5.4cm]{strings.4}\hfill
\includegraphics[width=4.4cm]{strings.5}\hfill~
\caption{Clause gadget for string graphs. Note that the vertex with maximum degree corresponds to the top green curve (left: representation, right: graph)}
\label{fig:stringclause}
\end{figure}

\section{Conclusion and Open Problems}
We have defined a new measure that can influence the hardness of the recognition problem for graph classes whose recognition is generally hard and we tried to get as tight results for particular graph classes as possible. For several classes, we have almost succeeded, usually 1 or 2 values are unclear. This fact motivates the open problems:

What is the complexity of recognizing for string graphs, segment graphs and GIGs (and preferably all classes between them) when we restrict our attention to graphs with maximum degree 3? How does the answer change when we restrict only to graphs with large girth (note that our reductions are meeting both criteria simultaneously)? We may also ask, what is the complexity of recognizing UGIGs with maximum degree 3 or 4 (both are open so far, our reduction works for degree at least 5). Further we may ask the same problem for string graphs with arbitrarily large girth and maximum degrees from 3 to 7.

For what graph classes does this new measure help more and which classes remain hard even with low degrees? What happens, e.g., with polygon-circle graphs (that can be polynomially recognized when restricted on graphs with girth at least 6) where the reduction requires vertices of high degree? It seems \cite{CG} that even any class between UGIGs and pseudosegment graphs remains NP-hard to get recognized even if the underlying graphs is almost planar (i.e., after removal of just one vertex it becomes planar) even for graphs with arbitrary girth (like our construction does). However, that construction requires one vertex to have a degree linearly large. What happens if we require all degrees to be bounded by a constant?

What restrictions do we have to apply for classes between GIGs and pseudosegments to make them polynomially recognizable? 

\bibliographystyle{acm}
\bibliography{ugig3}
\end{document}